\documentclass{article}

\usepackage[preprint]{neurips_2021}

\usepackage[utf8]{inputenc} 
\usepackage[T1]{fontenc}    
\usepackage{hyperref}       
\usepackage{url}            
\usepackage{booktabs}       
\usepackage{amsfonts}       
\usepackage{nicefrac}       
\usepackage{microtype}      
\usepackage{xcolor}         

\title{Conformal Bayesian Computation}
\author{%
  Edwin Fong \\
  University of Oxford\\
  The Alan Turing Institute\\
  \texttt{edwin.fong@stats.ox.ac.uk} \\
   \And
  Chris Holmes \\
  University of Oxford\\
  The Alan Turing Institute\\
  \texttt{cholmes@stats.ox.ac.uk} \\
}
\usepackage{amsmath}
\usepackage{amsthm}
\usepackage{algorithm2e}
\usepackage{xcolor}
\usepackage{bbm}

\def\T{{ \mathrm{\scriptscriptstyle T} }}

\usepackage{mathtools}

\newtheorem{proposition}{Proposition}
\newcommand{\iid}{\overset{\mathrm{iid}}{\sim}}

\begin{document}

\maketitle
\begin{abstract}
We develop scalable methods for producing conformal Bayesian predictive intervals with finite sample calibration guarantees. Bayesian posterior predictive distributions, $p(y \mid x)$,  characterize subjective beliefs on outcomes of interest, $y$, conditional on predictors, $x$. Bayesian prediction is well-calibrated when the model is true, but the predictive intervals may exhibit poor empirical coverage when the  model is misspecified, under the so called ${\cal{M}}$-open perspective. In contrast, conformal inference provides finite sample frequentist guarantees on predictive confidence intervals without the requirement of model fidelity. Using `add-one-in' importance sampling, we show that conformal Bayesian predictive intervals are efficiently obtained from re-weighted posterior samples of model parameters. Our approach contrasts with existing conformal methods that require expensive refitting of models or data-splitting to achieve computational efficiency. We demonstrate the utility on a range of examples including extensions to partially exchangeable settings such as hierarchical models.
\end{abstract}
\section{Introduction}

We consider Bayesian prediction using training data  $Z_{1:n} = \{Y_i, X_i\}_{i=1:n}$ for an outcome of interest  $Y_i$  and covariates $X_i \in \mathbb{R}^d$. Given a model likelihood $f_\theta(y \mid x)$ and prior on parameters, $\pi(\theta)$ for $\theta \in \mathbb{R}^p$, the posterior predictive distribution for the response at a new $X_{n+1} = x_{n+1}$ takes on the form
\begin{equation} \label{eq:post_pred}
p(y \mid x_{n+1},Z_{1:n}) = \int f_\theta(y \mid x_{n+1})\, \pi(\theta \mid Z_{1:n}) \, d\theta  \, , 
\end{equation}
where $\pi(\theta \mid Z_{1:n})$ is the Bayesian posterior. Asymptotically exact samples from the posterior can be obtained through Markov chain Monte Carlo (MCMC) and the above density can be computed through Monte Carlo (MC), or by direct sampling from an approximate model. 
Given a Bayesian predictive distribution, one can then construct the highest density $100 \times (1-\alpha)\%$ posterior predictive credible intervals, which  are the shortest intervals to contain  $(1-\alpha)$ of the predictive probability. Alternatively, the central $100\times(1-\alpha)\%$ credible interval can be computed using the ${\alpha}/{2}$ and $1-{\alpha}/{2}$ quantiles. Posterior predictive distributions condition on the observed $Z_{1:n}$ and represent subjective and coherent beliefs. 
However, it is well known that model misspecification can lead Bayesian intervals to be poorly \textit{calibrated} in the frequentist sense \citep{Dawid1982,Fraser2011}, that is the long run proportion of the observed data lying in the $(1 - \alpha)$ Bayes predictive interval is not necessarily equal to  $(1-\alpha)$. This has consequences for the robustness of such approaches and trust in using Bayesian models to aid decisions.

Alternatively, one can seek 
intervals around a point prediction from the model, $\widehat{y} = \widehat{\mu}(x)$, that have the correct frequentist coverage of $(1-\alpha)$ . This is precisely what is offered by the \textit{conformal prediction} framework of \citet{Vovk2005}, which allows the construction of prediction bands with finite sample validity without assumptions on the generative model beyond exchangeability of the data. Formally, for $Z_i = \{Y_i,X_i\}_{i=1:n}$, $Z_i \iid \mathbb{P}$ and miscoverage level $\alpha$, conformal inference allows us to construct a confidence set $C_\alpha(X_{n+1})$ from $Z_{1:n}$ and $X_{n+1}$ such that 
\begin{equation} \label{eq:coverage}
\mathbb{P}(Y_{n+1} \in C_\alpha(X_{n+1})) \geq 1-\alpha
\end{equation}
noting that $\mathbb{P}$ is over $Z_{1:n+1}$. In this paper we develop computationally efficient conformal inference methods for Bayesian models including extensions to hierarchical settings. A general theme of our work is that, somewhat counter-intuitively, Bayesian models are well suited for the conformal method.

Conformal inference for calibrating Bayesian models was previously suggested in \cite{Melluish2001}, \cite{Vovk2005}, \cite{Wasserman2011} and \cite{Burnaev2014}, where it is referred to as ``de-Bayesing'', ``frequentizing'' and ``conformalizing'', but only in the context of conjugate models.  Here, we present a scalable MC method for \textit{conformal Bayes},  implementing full conformal Bayesian prediction using an `add-one-in' importance sampling algorithm. The automated method can construct conformal predictive intervals from any Bayesian model given only samples of model parameter values from the posterior $ \theta \sim \pi(\theta \mid Z_{1:n})$, up to MC error. Such samples are readily available in most Bayesian analyses from probabilistic programming languages such as Stan \citep{Carpenter2017} and PyMC3 \citep{Salvatier2016}. We also extend conformal inference to partially exchangeable settings which utilize the important class of Bayesian hierarchical models, and note the connection to Mondrian conformal prediction \citep[Chapter~4.5]{Vovk2005}. Previously, the extension of conformal prediction to random effects was introduced in \cite{Dunn2020} in a non-Bayesian setting, with a focus on prediction in new groups, as well as within-group predictions without covariates. We will see that the Bayesian hierarchical model allows for a natural sharing of information between groups for within-group predictions with covariates. We discuss the motivation behind using the Bayesian posterior predictive density as the conformity measure for both the Bayesian and the frequentist, and demonstrate the benefits in a number of examples.

\subsection{Background}
The conformal inference framework was first introduced by \cite{Gammerman1998}, followed by the thorough book of \cite{Vovk2005}.  Full conformal prediction is computationally expensive, requiring the whole model to be retrained at each test covariate $x_{n+1}$ {\em{and}} for each value in a reference grid of potential outcomes, e.g. $y \in \mathbb{R}$ for regression.  This makes the task computationally infeasible beyond a few special cases where we can shortcut the evaluation along the outcome reference grid, e.g. ridge regression \citep{Vovk2005,Burnaev2014} and lasso \citep{Lei2019}. Shrinking the search grid is possible, but still requires many refittings of the model \citep{Chen2016}. The split conformal prediction method \citep{Lei2018} is a useful alternative method which only requires a single model fit, but increases variability by dividing the data into a training and test set that includes randomness in the choice of the split, and has a tendency for  wider intervals. Methods based on cross-validation such as cross-conformal prediction \citep{Vovk2015} and the jacknife+ \citep{Barber2021} lie in between the split and full conformal method in terms of computation. A detailed discussion of computational costs of various conformal methods are provided in \citet[Section 4]{Barber2021}. A review of recent advances in conformal prediction is given in \cite{Zeni2020}, and interesting extensions have been developed by works such as \citet{Tibshirani2019, Romano2019, Candes2021}.

\section{Conformal Bayes}
\subsection{Full Conformal Prediction}
We begin by summarizing the full conformal prediction algorithm discussed in \cite{Vovk2005,Lei2018}. Firstly, a conformity (goodness-of-fit) measure,
\begin{equation*}
\sigma_i:=\sigma(Z_{1:n+1}; Z_i),
\end{equation*}
takes as input a set of data points $Z_{1:n+1}$, and computes how similar the data point $Z_i$ is for $i = 1,\ldots, n+1$. A typical conformity measure for regression would be the negative squared error arising from a point prediction $-\left\{y_i - \widehat{\mu}(x_i)\right\}^2$, where $\widehat{\mu}(x)$ is the point predictor fit to the augmented dataset $Z_{1:n+1}$, assumed to be symmetric with respect to the permutation of the input dataset. The key property of any conformity measure is that it is exchangeable in the first argument, i.e. the conformity measure for $Z_i$ is invariant to the permutation of $Z_{1:n+1}$. Under the assumption that $Z_{1:n+1}$ is exchangeable, we then have that $\sigma_{1:n+1}$ is also exchangeable, and its rank is uniform among $\{1,\ldots,n+1\}$ (assuming continuous $\sigma_{1:n+1}$). From this, we have that the rank of $\sigma_{n+1}$ is a valid $p$-value. If we now consider a plug-in value  $Y_{n+1} = y$ (where $X_{n+1}$ is known), we can denote the rank of $\sigma_{n+1}$ among $\sigma_{1:n+1}$ as 
\begin{equation*}
\pi(y) = \frac{1}{n+1}\sum_{i=1}^{n+1} \mathbbm{1}\left(\sigma_i \leq \sigma_{n+1} \right).
\end{equation*} For miscoverage level $\alpha$, the full conformal predictive set,
\begin{equation}\label{eq:conformal_set}
C_\alpha(X_{n+1}) = \{y\in \mathbb{R}: \pi(y)>\alpha\},
\end{equation}
satisfies the desired frequentist coverage as in  (\ref{eq:coverage}).  {Intuitively, we are reporting the values of $y$ which conform better than the fraction $\alpha$ of observed conformity scores in the augmented dataset.} A formal proof can be found in \citet[Chapter~8.7]{Vovk2005}. For continuous $\sigma_{1:n+1}$, we also have from \citet[Theorem~1]{Lei2018} that the conformal predictive set does not significantly over-cover.

In practice, beyond a few exceptions, the function $\pi(y)$ must be computed on a fine grid $y\in \mathcal{Y}_{\text{grid}}$, for example of size 100, in which case   the model must be retrained 100 times to the augmented dataset to compute $\sigma_{1:n+1}$, with plug-in values for $y_{n+1}$ on the grid. This is illustrated in Algorithm \ref{alg:full_conformal} below. We note here that the grid method only provides approximate coverage, as $y$ may be selected even if it lies between two grid points that are not selected. This is formalized in \cite{Chen2018}, but we do not discuss this further. In the Appendix, we provide an empirical comparison of the grid effects. This is also valid for binary classification where we now have a finite $\mathcal{Y}_{\text{grid}} = \{0,1\}$, and so the grid method for full conformal prediction is exact and feasible.

\begin{center}
\begin{algorithm}[H]\label{alg:full_conformal}
\DontPrintSemicolon
  \SetAlgoLined
  Observed data is $Z_{1:n},X_{n+1}$; Specify miscoverge level $\alpha$\;
  \For{\textnormal{each} $y \in \mathcal{Y}_{\textnormal{grid}}$}{
    Fit model to augmented dataset $\{Z_{1},\ldots,Z_n,\{y,X_{n+1}\}\}$\;
    Compute $\sigma_{1:n}$ and $\sigma_{n+1} $\;
    Store the rank, $\pi(y)$ , of $\sigma_{n+1}$ among $\sigma_{1:n+1}$\;
    }
 Return the set $C_\alpha(X_{n+1}) = \{y\in \mathcal{Y}_{\text{grid}}: \pi(y) > \alpha\}$.\;
 \vspace{5mm}
\caption{Full Conformal Prediction}
\end{algorithm}
\end{center}

\subsection{Conformal Bayes and Add-One-In Importance Sampling}
In a Bayesian model, a natural suggestion for the conformity score, as noted in \cite{Vovk2005,Wasserman2011}, is the posterior predictive density (\ref{eq:post_pred}), that is
\begin{equation*}\label{eq:conformity_bayes}
\sigma(Z_{1:n+1}; Z_i) =  p(Y_i \mid X_i,Z_{1:n+1}).
\end{equation*}
This is a valid conformity score, as we have
$
\pi(\theta \mid Z_{1:n+1}) \propto \pi(\theta) \prod_{i=1}^{n+1} f_\theta(Y_i \mid X_i),
$ and so $\sigma$ is indeed invariant to the permutation of $Z_{1:n+1}$. We denote this method as \textit{conformal Bayes} (CB), and we will see shortly that the exchangeability structure of Bayesian models is key to constructing conformity scores in the partial exchangeability scenario.

Beyond conjugate models, we are usually able to obtain (asymptotically exact) posterior samples $\theta^{(1:T)} \sim \pi(\theta \mid Z_{1:n})$, e.g through MCMC, where $T$ is a large integer. Such samples are typically available as standard output from Bayesian model fitting. The posterior predictive can then be computed up to Monte Carlo error through
\begin{equation*}
\widehat{p}(Y_{n+1} \mid X_{n+1},Z_{1:n}) = \frac{1}{T}\sum_{t = 1}^T f_{\theta^{(t)}}(Y_{n+1} \mid X_{n+1}).
\end{equation*}

The key insight is that refitting the Bayesian model with $\{Z_{1},\ldots,Z_n,\{y,X_{n+1}\}\}$ is well approximated through importance sampling (IS), as only $\{y,X_{n+1}\}$ changes between refits. This leads immediately to an IS approach to full conformal Bayes, where we just need to compute `add-one-in' (AOI) predictive densities. Here AOI refers to the inclusion of $\{Y_{n+1}, X_{n+1}\}$ into the training set, named in relation to  `leave-one-out' (LOO) cross-validation. Specifically, for $Y_{n+1} = y$ and  $\theta^{(1:T)} \sim \pi(\theta \mid Z_{1:n})$, we can compute
\begin{equation}\label{eq:IS_sum}
\widehat{p}(Y_i \mid X_i,Z_{1:n+1}) =\sum_{t=1}^T \widetilde{w}^{(t)}f_{\theta^{(t)}}(Y_i \mid X_i)
\end{equation}
where $\widetilde{w}^{(t)}$ are our self-normalized importance weights of the form
\begin{equation}\label{eq:IS_weights}
\begin{aligned}
w^{(t)} &= f_{\theta^{(t)}}(y \mid X_{n+1}), \quad
\widetilde{w}^{(t)} &= \frac{w^{(t)}}{\sum_{t'=1}^T w^{(t')}} \cdot
\end{aligned}
\end{equation}
We see that the unnormalized importance weights have the intuitive form of the predictive likelihood at the reference point $\{y,X_{n+1}\}$ given the model parameters $\theta^{(t)}$.

The use of AOI importance sampling has similarities to  the computation of Bayesian leave-one-out cross-validation (LOOCV) predictive densities \citep{Vehtari2017}, which is also used in accounting for model misspecification. 
An interesting aspect of AOI in comparison with LOO is that AOI predictive densities are less vulnerable to importance weight instability for the following reasons:
\begin{itemize}
    \item In LOOCV, the target $\pi(\theta \mid Z_{-i})$ generally has thinner tails than the proposal $\pi(\theta \mid Z_{1:n})$, leading to importance weight instability. In contrast, AOI uses the posterior $\pi(\theta \mid Z_{1:n})$ as a proposal for the thinner-tailed $\pi(\theta \mid Z_{1:n+1})$. For LOOCV the importance weights are proportional to $1/f_\theta(y \mid x)$, in contrast to the typically bounded $f_{\theta}(y \mid x)$ for AOI.
    \item For AOI, we are predicting $Z_i$ given $Z_{1:n+1}$ which is always in-sample unlike in LOOCV where the datum is out-of-sample, so we can expect greater stability with AOI. 
    \item The IS weight stability is governed by $Y_{n+1} = y$, which is not random as we select it for the grid. For sufficiently large $\alpha$, we will not need to compute the AOI predictive density for extreme values of $y$. 
    \end{itemize}
    
    We provide some IS weight diagnostics in the experiments and find that they are stable. In difficult settings such as very high-dimensions, one can make use of the recommendations of \citet{Vehtari2015} for assessing and Pareto-smoothing the importance weights if necessary.

\subsection{Computational complexity}

Given the posterior samples,
we must compute the likelihood for each $\theta^{(t)}$ at $Z_{1:n}$, as well at $\{y,X_{n+1}\}$ for $y\in \mathcal{Y}_{\text{grid}}$. The additional computation required for CB for each $X_{n+1}$ is thus $T \times (n  + n_{\text{grid}})$ likelihood evaluations, which is relatively cheap. This is then followed by the dot product of an $(n+1) \times T$ matrix with a $T$ vector for each $y$, which is $\mathcal{O}(nT)$, so the overall  complexity is $\mathcal{O}(n_{\text{grid}}Tn)$. The values $n_{\text{grid}}$ and $T$ are constants, though we may want to increase $T$ with the dimensionality of the model to reduce importance sampling variance. The large matrices involved in computing the AOI predictives suggests we can take advantage of GPU computation, and machine learning packages such as JAX \citep{Jax2018} are highly suitable for this application.

\subsection{Motivation}\label{sec:motivation}
Much has been written on the contrasting foundations and interpretation of Bayes versus frequentist measures of uncertainty \citep{Little2006,Shafer2008,Bernardo2009, Wasserman2011}, and we provide a summary in the Appendix. Here we motivate CB predictive intervals from both a Bayesian and frequentist perspective.

The pragmatic Bayesian, aware of the potential for model misspecification in either the prior or likelihood, may be interested in conformal inference as a countermeasure. CB predictive intervals with guaranteed frequentist coverage can be provided as a supplement to the usual Bayesian predictive intervals. The difference between the Bayesian and conformal interval may also serve as an informal diagnostic for model evaluation (e.g. \cite{Gelman2013}). Posterior samples through MCMC or direct sampling are typically available, and so CB through automated AOI carries little overhead. 

The frequentist may also wish to use a Bayesian model as a tool for constructing predictive confidence intervals. Firstly, the likelihood can take into account skewness, heteroscedasticity unlike the usual residual conformity score.  Secondly, features such as sparsity, support, and regularization can be  included through priors, while CB ensures correct coverage. Finally, a subtle issue that arises in full conformal prediction is that we lose validity if hyperparameter selection is not symmetric with respect to $Z_{n+1}$, e.g. if we estimate the lasso penalty $\lambda$ using only $Z_{1:n}$ before computing the full conformal intervals with said $\lambda(Z_{1:n})$. For CB, a prior on hyperparameters induces weighting of the hyperparameter values by implicit cross-validation for each refit \citep{Gneiting2007, Fong2020}. We highlight here that this issue does not affect the split conformal method.

\section{Partial Exchangeability and Hierarchical Models}

A setting of particular interest is for grouped data, which corresponds to a weakening of exchangeability often denoted as partial exchangeability \citep[Chapter~4.6]{Bernardo2009}. Assume that we observe data from $J$ groups, each of size $n_j$, where again $Z_{i,j} = \{Y_{i,j},X_{i,j}\}$. We denote the full dataset as $Z= \{Z_{i,j}: i = 1,\ldots, n_j,  j = 1,\ldots, J\}$. We may not expect the entire sequence $Z$ to be exchangeable, instead only that data points are exchangeable within groups. Formally, this means that
\begin{equation}\label{eq:part_exch}
p(Z_{1:n_1,1},\ldots, Z_{1:n_J,J})=p(Z_{\pi_1(1):\pi_1(n_1), 1}, \ldots,Z_{\pi_J(1):\pi_J(n_J), J})
\end{equation}
for any permutations $\pi_j$ of $1,\ldots,n_j$, for $j = 1,\ldots,J$. Alternatively, we can enforce the usual definition of exchangeability but only consider permutations $\pi$ of $1,\ldots,n$ such that the groupings are preserved. A simple example of this partial exchangeability is if $Z_{i,j} \iid P_j$ for $i =1 ,\ldots,n_j,\, j =1,\ldots,J$, where $P_j$ can now be distinct.

Partial exchangeability is useful in  multilevel modelling, e.g. where $Z_{1:n_j,j}$ records exam results on students within school $j$, for schools $j = 1,\ldots,J$. Students may be deemed exchangeable within schools, but not between schools. Further examples may be found in \citet{Gelman2006}.

\subsection{Group Conformal Prediction}\label{sec:group_conf}

Given a new $X_{n_j+1,j}$ belonging to group $j$ for $j \in \{1,\ldots,J\}$, we seek to construct a $(1-\alpha_j)$ confidence interval for $Y_{n_j+1,j}$. We define a within-group conformity score as 
\begin{equation*}
\sigma_{i,j} := \sigma_{Z_{-j}}(Z_{1:n_j+1,j}; Z_{i,j})
\end{equation*}
for $i = 1,\ldots, n_j+1$. We denote $Z_{-j}$ as the dataset without group $j$, and the subscript indicates the dependence of the conformity score on this, which we motivate in the next subsection. For each $Z_{-j}$, we require the score to be invariant with respect to the permutation of $Z_{1:n_j+1,j}$. For $Z_{n_j+1,j} = \{y,X_{n_j+1,j}\}$, the conformal predictive set is then defined 
\begin{equation}\label{eq:conformal_set_group}
\begin{aligned}
\pi_j(y) = \frac{1}{n_j+1}\sum_{i=1}^{n_j+1} \mathbbm{1}\left(\sigma_{i,j} \leq \sigma_{n_j+1,j} \right),\quad C_{\alpha_j}\left(X_{n_j+1,j}\right) = \{y \in \mathbb{R}: \pi_j(y) > \alpha_j)
\end{aligned}
\end{equation}
In other words, we rank the conformity scores $\sigma_{1:n_j+1,j}$ within the group $j$, and compute the conformal interval as usual with Algorithm \ref{alg:full_conformal}. The interval is valid from the following.

\begin{proposition}
Assume that $\{Z, Z_{n_j+1,j}\}$ is partially exchangeable as in (\ref{eq:part_exch}), and the conformity measure $\sigma_{i,j}$ for group $j$ is invariant to the permutation of $Z_{1:n_j+1,j}$. We then have
\begin{equation*}
\mathbb{P}\left(Y_{n_j+1,j} \in C_{\alpha_j}\left(X_{n_j+1,j} \right) \right)\geq 1-\alpha_j
\end{equation*}
where $ C_{\alpha_j}\left(X_{n_j+1,j} \right)$ is defined in (\ref{eq:conformal_set_group}), and $\mathbb{P}$ is over $\{Z,Z_{n_j+1,j}\}$.
\end{proposition}
\begin{proof}
Conditional on $Z_{-j}$, the observations $Z_{1:n_j +1,j}$ are still exchangeable, and thus so are $\sigma_{1:n_j+1,j}$ from the invariance of the conformity measure. The usual conformal guarantee then holds:
\begin{equation*}
\mathbb{P}\left(Y_{n_j+1,j} \in C_{\alpha_j}\left(X_{n_j+1,j} \right) \mid Z_{-j} \right)\geq 1-\alpha_j.
\end{equation*}

Taking the expectation with respect to $Z_{-j}$ gives us the result.
\end{proof}

It is interesting to note that the above group conformal predictor coincides with the attribute-conditional Mondrian conformal predictor of \citet[Chapter~4.5]{Vovk2005}, with the group allocations as the taxonomy. Validity under the relaxed Mondrian-exchangeability of \citet[Chapter~8.4]{Vovk2005} is key for us here.

\subsection{Conformal Hierarchical Bayes}
Under this setting, a hierarchical Bayesian model can be defined of the form
\begin{equation*}
    \begin{aligned}
    [Y_{i,j} \mid X_{i,j},\theta_j,\tau] &\iid f_{\theta_j,\tau}(\cdot \mid X_{i,j}) \quad i = 1,\ldots, n_j, \quad &j = 1,\ldots, J \\
    [\theta_j \mid \phi]&\iid \pi(\cdot \mid \phi) \quad &j = 1,\ldots, J\\
    \phi \sim \pi(\phi), &\quad \tau \sim \pi(\tau).
    \end{aligned}
\end{equation*}
Here $\tau$ is a common parameter across groups (e.g. a common standard deviation for the residuals under homoscedastic errors).  The desired partial exchangeability structure is clearly preserved in the Bayesian model \citep{Bernardo1996}. De Finetti representation theorems are also available for partially exchangeable sequences (when defined in a slightly different manner to the above), which motivate the specification of hierarchical Bayesian models \citep[Chapter~4.6]{Bernardo2009}.

The posterior predictive is once again a natural choice for the conformity measure. Denoting $\bar{Z}_y$ as the entire dataset augmented with $Z_{n_j+1,j} = \{y,X_{n_j+1,j}\}$, we have
\begin{equation}\label{eq:group_posterior_pred}
\sigma_{i,j} = p(Y_{i,j} \mid X_{i,j}, \bar{Z}_y) = \int f_{\theta_j,\tau}(Y_{i,j} \mid X_{i,j})\,  \pi(\theta_j,\tau \mid \bar{Z}_y) \, d\theta_j \, d\tau 
\end{equation}
for $i = 1,\ldots,n_j+1$. The within-group permutation invariance follows as the likelihood is exchangeable within groups, and thus so is the posterior and resulting posterior predictive. Practically, this structure allows for independent coefficients $\theta_j$ for each group, but partial pooling through $\pi( \theta \mid \phi)$ allows information to be shared between groups. A fully pooled model, whilst still valid, is usually too simple and predicts poorly, whereas a no-pooling conformity score ignores information sharing between groups. More details on hierarchical models can be found in \citet[Chapter~5]{Gelman2013}. We point out that we can select a separate coverage level $\alpha_j$ for each group, which will be useful when group sizes $n_j$ vary - we provide a demonstration of this in the Appendix. Computation of $\sigma_{i,j}$ is again straightforward, where MCMC now returns $[\theta_{1:J}^{(1:T)},\phi^{(1:T)},\tau^{(1:T)}] \sim \pi(\theta_{1:J},\phi,\tau \mid Z)$. We can then estimate (\ref{eq:group_posterior_pred}) using AOI importance sampling as in (\ref{eq:IS_sum}) and (\ref{eq:IS_weights}) using the marginal samples $\{\theta_j^{(1:T)},\tau^{(1:T)}\} \sim \pi(\theta_j,\tau \mid Z)$ and weights $w^{(t)}  = f_{\theta_j^{(t)}, \tau^{(t)}}(y \mid X_{n_j+1,j})$.

In the above, we consider predictive intervals within groups with covariates, extending the within-group approach of \cite{Dunn2020}. Predictive intervals for new groups are possible with the Bayesian model, but a conformal predictor would require additional stronger assumptions of exchangeability to ensure validity. The relevant assumptions and methods based on pooling and subsampling for new group predictions are discussed in \citet{Dunn2020}, but would require rerunning MCMC in our case. We leave this for future work, noting that utilizing the Bayesian predictive density directly here seems nontrivial due to different group sizes.

\section{Experiments}\label{sec:experiments}
We run and time all examples on an Azure NC6 Virtual Machine, which has 6 Intel Xeon E5-2690 v3 vCPUs and a one-half Tesla K80 GPU card. We use PyMC3 \citep{Salvatier2016} for MCMC and \texttt{sklearn} \citep{Pedregosa2011} for the regular conformal predictor; both are run on the CPU. Computation of the CB and Bayes intervals is implemented in JAX \citep{Jax2018}, and run on the GPU.  The code is available online\footnote{\href{https://github.com/edfong/conformal_bayes}{\url{https://github.com/edfong/conformal_bayes}}} and further examples are  provided in the Appendix.

\subsection{Sparse Regression}\label{sec:sparse_reg}
We first demonstrate our method under a sparse linear regression model on the diabetes dataset  \citep{Efron2004} considered by \citet{Lei2019}. The dataset is available in \texttt{sklearn}, and consists of $n = 442$ subjects, where the response variable is a continuous diabetes progression and the $d = 10$ covariates consist of patient readings such as blood serum measurements. We standardize all covariates and the response to have mean 0 and standard deviation 1.

The Bayesian model we consider is
\begin{equation}\label{eq:priors}
\begin{aligned}
    f_\theta(y \mid x) = \mathcal{N}(y \mid \theta^\T x + \theta_0, \tau^2)\\
    \pi(\theta_j) = \text{Laplace}(0,b), \quad \pi(\theta_0) \propto 1, \quad \pi(b) = \text{Gamma}(1,1)&\quad  \pi(\tau)  = \mathcal{N}^+(0,c)
    \end{aligned}
\end{equation}
for $j = 1,\ldots,d$, and where $b$ is the scale parameter and $\mathcal{N}^+$ is the half-normal distribution. Note that a hyperprior on $b$ has removed the need for cross-validation that is required for lasso. We consider two values of $c$ for the hyperprior on $\tau$, which correspond to a well-specified ($c = 1$) and poorly-specified $(c = 0.02)$ prior; in the latter case our posterior on $\tau$ will be heavily weighted towards a small value. This model is well-specified for the diabetes dataset \citep[Chapter~4.5]{Jansen2013} under a reasonable prior ($c = 1$). We compute the central $(1-\alpha)$ credible interval from the Bayesian posterior predictive CDF estimated using Monte Carlo and the same grid as for CB. 

To check coverage, we repeatedly divide into a training and test dataset for $50$ repeats, with $30\%$ of the dataset in the test split. We evaluate the conformal prediction set on a grid of size $n_{\text{grid}} = 100$ between $[y_{\text{min}} -2,y_{\text{max}}+2]$, where $y_{\text{min}},y_{\text{max}}$ is computed from each training dataset. The average coverage, length and run-times (\textit{excluding} MCMC) with standard errors are given in Table \ref{tab:diab} for $\alpha = 0.2$. MCMC induced an average overhead of $21.9$s for $a = 1$ and $26.8$s for $c = 0.02$ for the Bayes and CB interval, where we simulate $T = 8000$ posterior samples. The CB intervals are only slightly slower than the Bayes intervals, and still a small fraction of the time required for MCMC, and is thus an efficient  post-processing step. For $c = 1$, the Bayesian intervals have coverage close to $(1-\alpha)$ with the smallest expected length, with CB slightly wider and more conservative. However, when the prior is  misspecified with $c = 0.02$, the Bayes intervals severely undercover, whilst the CB coverage and length remain unchanged from the $c = 1$ case.  

As baselines, we compare to the split and full conformal method using the non-Bayesian lasso as the predictor, with the usual residual as the nonconformity score.  For the split method, we fit lasso with cross-validation on the subset of size  $n_{\text{train}}/2$ to obtain the lasso penalty $\lambda$. For the full conformal method, we use the grid method for fair timing, as other estimators beyond lasso would not have the shortcut of \cite{Lei2019}. As setting a default $\lambda = 1$ gives poor average lengths, we estimate $\lambda = 0.004$ on cross-validation on one of the training sets, and use this value over the 50 repeats. However, we must emphasize again that this is somewhat misleading, as discussed in Section \ref{sec:motivation}. A fairer approach would involve fitting lasso with CV for each of the 100 grid values and 133 test values, but this is infeasible as each fit requires around 80ms, resulting in a total run-time of 17 minutes. On the other hand, the AOI scheme of CB is equivalent to refitting $b$ for each grid/test value. In terms of performance, the split method has wider intervals than CB/full, but performs well given the low computational costs. The full conformal method performs as well as CB, but is comparable in time as MCMC + CB, whilst not refitting $\lambda$. We note that the value of $c$ does not affect the split/full method.

\begin{table}[!h]
\begin{center} 
\caption{Diabetes; Coverage values \textit{not} within 3 standard errors (in brackets) of the target coverage $(1-\alpha) = 0.8$ are in {\color{red}\textbf{red}}.}
\label{tab:diab}
\begin{tabular}{c c c c c c}
\hline
&& Bayes & CB & Split&Full ($\lambda =0.004$)\\
\hline 
 Coverage &$c = 1$& {0.806} (0.005) & {0.808} (0.006) &  0.816 (0.006)&{0.808} (0.006)\\ 
 &$c = 0.02$&{\color{red}{\textbf{0.563}}} (0.006) & {0.809} (0.006) &  /&/\\
    \hline
 Length &$c = 1$& 1.84 (0.01) & 1.87 (0.01) & 1.95 (0.02)&1.86 (0.01)\\ 
 &$c = 0.02$& 1.14 (0.00) & 1.87 (0.01) & /& /\vspace{0.5mm}\\
    \hline
 \rule{0pt}{0.8\normalbaselineskip} Run-time & $c = 1$& 0.488 (0.107)& 0.702 (0.019)&   0.065 (0.001)&11.529 (0.232)\\
(secs)  &$c =0.02$ &0.373 (0.002)&0.668 (0.003)&  /&/\\
\end{tabular}
\end{center}
\end{table}\vspace{-2mm}

\subsubsection{Importance weights}
For the diabetes dataset, we look at the effective sample size (ESS) of the self-normalized importance weights (\ref{eq:IS_weights}), which can be computed as $\text{ESS}= 1/\sum_{t = 1}^T \{w^{(t)}\}^2$ for each $x_{n+1}$ and $y$. The ESS as a function of $y$ for a single $x_{n+1}$ is shown in Figure \ref{fig:ESS} for the two cases $c = 1,0.02$, with the CB conformal bands given for $\alpha = 0.2,0.5$. We have scaled the ESS plots by $\text{ESS}_{\text{MCMC}}/ T$, where $T=8000$ is the number of posterior samples and $\text{ESS}_{\text{MCMC}}$ is the minimum ESS out of all posterior parameters return by PyMC3.  We observe the ESS is well behaved and stable across the range of $y$ values.  In both cases, the ESS for $\alpha = 0.2$ is sufficiently large for a reliable estimate of the conformity scores. However, for $c = 0.02$, the ESS decays more quickly with $y$ as the Bayes predictive intervals are too narrow, which the CB corrects for. Other values of $x_{n+1}$ produce similar behaviour.

 \begin{figure}[!h]
\centering
 \includegraphics[width=0.9\textwidth]{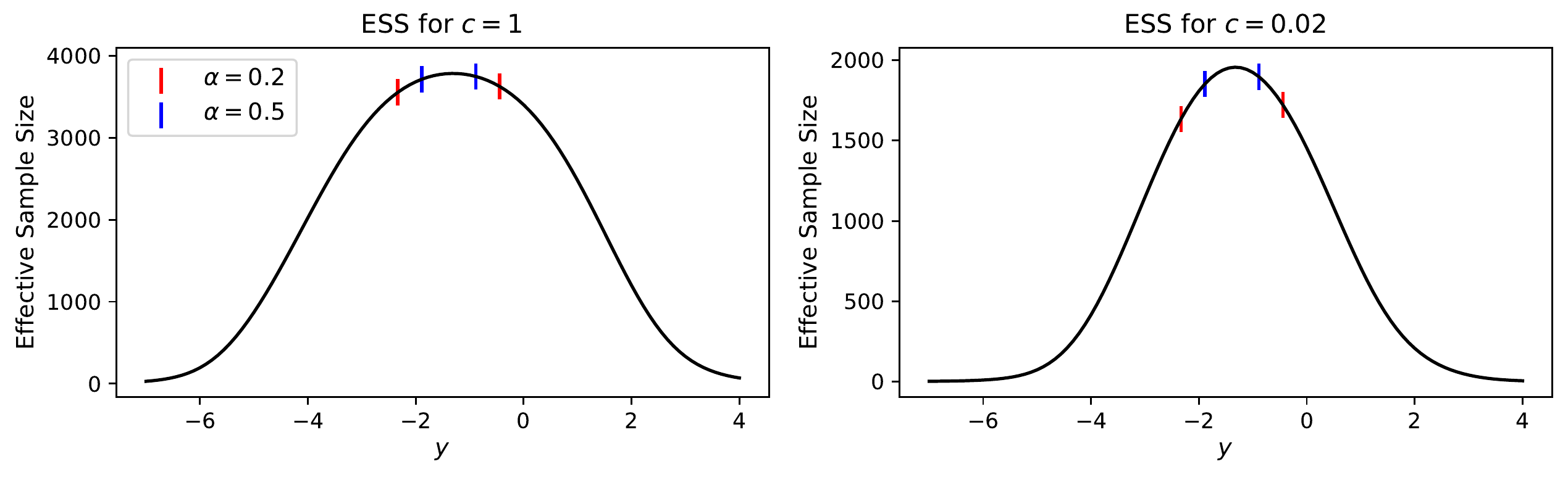}
\caption{Effective sample sizes of IS weights with CB conformal bands for diabetes dataset with (left) $c = 1$ and (right) $c = 0.02$.}\label{fig:ESS}
\end{figure}

\subsection{Sparse Classification}\label{sec:sparse_class}
In this section, we analyze the Wisconsin breast cancer \citep{Wolberg1990}, again available in \texttt{sklearn}. The dataset is of size $569$, where the binary response variable corresponds to a malignant or benign tumour. The $30$ covariates consist of measurements of cell nuclei. Again, we standardize all covariates to have mean 0 and standard deviation 1.

We consider the logistic likelihood
$
f_\theta(y=1 \mid x) = [1+ \exp\left\{-\left(\theta^\T x + \theta_0 \right) \right\}]^{-1},
$ 
with the same priors for $\theta,\theta_0$ as in (\ref{eq:priors}). The Bayesian predictive set is the smallest set from $\{0\},\{1\},\{0,1\}$ that contains at least $(1-\alpha)$ of the posterior predictive probability. The conformal baselines are as above but with $L_1$-penalized logistic regression, and for the full conformal method we have $\lambda = 1$. We again have 50 repeats with 70-30 train-test split, and set $\alpha = 0.2$. The grid method is now exact, and the size of the CB intervals can take on the values $\{0,1,2\}$. The results are provided in Table \ref{tab:breast_cancer}, where MCMC required an average of $45.4$s to produce $T = 8000$ samples.  We see that even with reasonable priors, Bayes can over-cover substantially, which CB corrects in roughly the same amount of time as it takes to compute the usual Bayes interval. 
However, we point out that CB may produce empty prediction sets, whereas Bayes cannot, and we investigate this in the Appendix.

\begin{table}[!h]
\begin{center}
\caption{Breast Cancer;  Coverage values \textit{not} within 3 standard errors (in brackets) of the target coverage $(1-\alpha) = 0.8$ are in {\color{red}\textbf{red}}. ``Size'' denotes the average number of elements in the conformal prediction set, averaged over the test points and repetitions.}\label{tab:breast_cancer}
\begin{tabular}{c c c c c}
\hline
& Bayes & CB &  Split& Full\\
\hline 
 Coverage & {\color{red}{\textbf{0.990}}} (0.001) & 0.812 (0.005)& {0.814} (0.006)&{0.811} (0.005)\\
    \hline \vspace{0.5mm}
 Size & 1.06 (0.00) & 0.81 (0.00) &0.82 (0.01)&0.81 (0.00)\vspace{0.5mm}\\
    \hline \vspace{0.5mm}  
  \rule{0pt}{0.8\normalbaselineskip} Run-time (secs) &  0.364 (0.007)& 0.665 (0.012)& 0.079 (0.002)& 1.008 (0.016)\\
\end{tabular}
\end{center}
\end{table}\vspace{-2mm}

\subsection{Hierarchical Model}
We now demonstrate Bayesian conformal inference using a hierarchical Bayesian model for multilevel data. We stick to the varying intercept and varying slope model \citep{Gelman2013}, that is for $j = 1,\ldots,J$:
\begin{equation}\label{eq:hier_model}
\begin{aligned}
    f_{\theta_j,\tau}(y_{i,j}) = \mathcal{N}(y_{i,j} \mid \theta^\T_j X_{i,j} +\theta_{0,j} , \tau^2)\\
    \pi(\theta_j) = \mathcal{N}(\phi, s^2), \quad \pi(\theta_{0,j}) = \mathcal{N}(\phi_0, s_0^2) \quad
    \\
    \end{aligned}
\end{equation}
with hyperpriors $\mathcal{N}(0,1)$ on the location parameters $\phi,\phi_0$ and $\text{Exp}(1)$ on the standard deviations $s,s_0,\tau$. We now apply this to a simulated example, and an application to the radon dataset of \cite{Gelman2006} is given in the Appendix.

We consider two simulation scenarios, with $J = 5$ groups and  $n_j = 10$ elements per group:
\begin{enumerate}
    \item Well-specified: We generate group slopes $\theta_{j} \iid \mathcal{N}(0,1)$ for $j = 1,\ldots, J$. For each $j$, we generate $X_{i,j} \sim \mathcal{N}(0,1)$ and $Y_{i,j} \sim \mathcal{N}\left(\theta_j X_{i,j},1 \right)$.
    \item Misspecified: We generate group slopes and variances $\theta_{j} \iid \mathcal{N}(0,1), \tau_j \iid \text{Exp}(1)$ for $j = 1,\ldots, J$. For each $j$, we generate $X_{i,j} \sim \mathcal{N}(0,1)$ and $Y_{i,j} \sim \mathcal{N}\left(\theta_j X_{i,j},\tau_j^2 \right)$.
\end{enumerate}
The first scenario has homoscedastic noise between groups as assumed in the model (\ref{eq:hier_model}) whereas the second scenario is heteroscedastic between groups.  To evaluate coverage, we only draw $\theta_{1:J},\tau_{1:J}$ once (and not per repeat), giving us the values
$$
\theta_{1:J} = [ 1.33, -0.77, -0.32, -0.99, -1.07], \quad \tau_{1:J} = [1.24, 2.30, 0.76, 0.28, 1.11].
$$

For each of the $50$ repeats, we draw $n_j = 10$ training and test data points from each group using the above $\theta_{1:J}$ (and $\tau_{1:J}$ for scenario 2), and report test coverage and lengths within each group. We use a grid of size 100 between $[-10,10]$. The group-wise average lengths and coverage are given in Table \ref{tab:hier_sim} again with $\alpha = 0.2$. Again run-times are given post-MCMC, where MCMC required an average of $90.1$s and $78.4$s for scenarios 1 and 2 respectively to generate $T = 8000$ samples. The Bayes interval is again the central $(1-\alpha)$ credible interval. The CB and Bayes methods have comparable run-times, likely due to the small $n$. As a reference, fitting a linear mixed-effects model in \texttt{statsmodels} \citep{Seabold2010} to the dataset takes around 200ms, so the full conformal method, which requires refitting for each of the 100 grid value and 50 test values, would take a total of 17 minutes. For scenario 1, both Bayes and CB provide close to $(1-\alpha)$ coverage, with the Bayes lengths being smaller. This is unsurprising, as the Bayesian model is well-specified. In scenario 2, the Bayes intervals noticeably over/under-cover depending on the value of $\tau_{1:J}$ in relation to the Bayes posterior mean $\bar{\tau} \approx 1.3$. CB is robust to this, adapting its interval lengths accordingly (in particular for Groups 2 and 4) and providing within-group validity.

\begin{table}[!h]
\begin{center}
\caption{Simulated grouped dataset;  Coverage values \textit{not} within 3 standard errors (in brackets) of the target coverage $(1-\alpha) = 0.8$ are in {\color{red}{\textbf{red}}}.}
\label{tab:hier_sim}
\begin{tabular}{c c c c c  c c}
&&\multicolumn{2}{c}{Scenario 1}&\multicolumn{2}{c}{Scenario 2}\\
\hline
&Group& Bayes & CB &  Bayes & CB\\
\hline \vspace{0.5mm}
 Coverage&1 &0.808 (0.020)  & 0.794 (0.022)&0.826 (0.020)  & 0.786 (0.025) \\ 
    &2 & 0.800 (0.019)&  0.812 (0.024)&{\color{red}{\textbf{0.522}}} (0.027)&  0.812 (0.024) \\ 
    &3 & 0.824 (0.017)& 0.824 (0.022)&{\color{red}{\textbf{0.974}}} (0.008)& 0.824 (0.020) \\ 
    &4 &0.786 (0.017)& 0.798 (0.022)&{\color{red}{\textbf{1.000}}} (0.000)& 0.836 (0.021) \\ 
    &5 &0.772  (0.019)& 0.810 (0.020)&0.826  (0.022)& 0.796 (0.022)  \vspace{0.5mm}\\ 
    & Overall & 0.798 (0.009)& 0.808 (0.009)&  {\color{red}{\textbf{0.830}}} (0.010)& 0.811 (0.009)\vspace{0.5mm}\\
    \hline \vspace{0.5mm}
 Length&1 & 2.80 (0.05)& 3.19 (0.13)& 3.65 (0.08)& 4.01 (0.17)\\ 
    &2 & 2.76 (0.05)& 3.21 (0.15)& 3.61 (0.08)& 7.27 (0.33) \\ 
    &3 & 2.75 (0.04)&3.07 (0.13)& 3.59 (0.08)& 2.28 (0.09) \\ 
    &4 &  2.75 (0.05)& 3.05 (0.12)&3.57 (0.08)& 1.23 (0.04)  \\ 
    &5 &  2.78 (0.05)& 3.14 (0.11)&3.61 (0.08)& 3.47 (0.12) \vspace{0.5mm}\\ 
    & Overall &2.77 (0.04)& 3.13 (0.06) &3.61(0.08)& 3.65 (0.09)  \vspace{0.5mm}\\
    \hline\vspace{0.5mm}
    \rule{0pt}{0.8\normalbaselineskip}  Run-time (secs) &Overall&0.222 (0.002) & 0.381 (0.009)& 0.221 (0.002)& 0.375 (0.002)
\end{tabular}
\end{center}
\end{table}\vspace{-5mm}

\section{Discussion}\label{sec:discussion}

In this work, we have introduced the AOI importance sampling scheme for conformal Bayesian computation, which allow us to construct frequentist-valid predictive intervals from a baseline Bayesian model using the output of an MCMC sampler. This extends naturally to the partially exchangeable setting and hierarchical Bayesian models.

Under model misspecification, or the $\mathcal{M}$-open scenario \citep{Bernardo2009}, CB can produce calibrated intervals from the Bayesian model. In the partially exchangeable case, CB can remain valid within groups. We find that even under reasonable priors, Bayesian predictives can over-cover, and CB can help reduce the length of intervals to get closer to nominal coverage. Diagnosing Bayesian miscalibration is in general non-trivial, but CB automatically corrects for this. When posterior samples of model parameters are available, AOI importance sampling is only a minor increase in computation, and interestingly is much faster than the split method which would require another run of MCMC. For the frequentist, CB intervals enjoy the tightness of the full conformal method, for a single expensive fit with MCMC followed by a cheap refitting process. We are also free to incorporate prior information, and use more complex likelihoods or priors, as well as automatically fitting hyperparameters. 

There are however limitations to our approach, dictated by the realities of MCMC and IS. Firstly, the intervals are approximate up to MC error and reliant on representative MC samples not disrupting exchangeability of the conformity scores. The stability of AOI importance sampling also depends on the posterior predictive being a good proposal, which may break down if the addition of the new datum $\{y,X_{n+1}\}$ has very high leverage on the posterior. 

If only approximate posterior samples are available, e.g. through variational Bayes (VB), then an AOI scheme may still be feasible, where one includes an additional correction term in the IS weights for the VB approximation, e.g. in \cite{Magnusson2019}. However, this remains to be investigated. Combining this with the Pareto-smoothed IS method of \citet{Vehtari2015} may lead to additional scalability with dimensionality. In our experience, CB intervals tend to be a single connected interval, which may allow for computational shortcuts in adapting the search grid. It would also be interesting to pursue the theoretical connections between the Bayesian and CB intervals, in a similar light to \cite{Burnaev2014}.

\section*{Acknowledgements}
Fong is funded by The Alan Turing Institute Doctoral Studentship, under the EPSRC grant EP/N510129/1. Holmes is supported by The Alan Turing Institute, the Health Data Research, U.K., the Li Ka Shing Foundation, the
Medical Research Council, and the U.K. Engineering and Physical
Sciences Research Council through the Bayes4Health programme grant.

\bibliographystyle{apalike}
\bibliography{refs}

\begin{thebibliography}{}

\bibitem[Barber et~al., 2021]{Barber2021}
Barber, R.~F., Candes, E.~J., Ramdas, A., Tibshirani, R.~J., et~al. (2021).
\newblock Predictive inference with the jackknife+.
\newblock {\em Annals of Statistics}, 49(1):486--507.

\bibitem[Bernardo, 1996]{Bernardo1996}
Bernardo, J.~M. (1996).
\newblock The concept of exchangeability and its applications.
\newblock {\em Far East Journal of Mathematical Sciences}, 4:111--122.

\bibitem[Bernardo and Smith, 2009]{Bernardo2009}
Bernardo, J.~M. and Smith, A.~F. (2009).
\newblock {\em Bayesian theory}, volume 405.
\newblock John Wiley \& Sons.

\bibitem[Bradbury et~al., 2018]{Jax2018}
Bradbury, J., Frostig, R., Hawkins, P., Johnson, M.~J., Leary, C., Maclaurin,
  D., and Wanderman-Milne, S. (2018).
\newblock {JAX}: composable transformations of {P}ython+{N}um{P}y programs.

\bibitem[Burnaev and Vovk, 2014]{Burnaev2014}
Burnaev, E. and Vovk, V. (2014).
\newblock Efficiency of conformalized ridge regression.
\newblock In {\em Conference on Learning Theory}, pages 605--622.

\bibitem[Cand{\`e}s et~al., 2021]{Candes2021}
Cand{\`e}s, E.~J., Lei, L., and Ren, Z. (2021).
\newblock Conformalized survival analysis.
\newblock {\em arXiv preprint arXiv:2103.09763}.

\bibitem[Carpenter et~al., 2017]{Carpenter2017}
Carpenter, B., Gelman, A., Hoffman, M.~D., Lee, D., Goodrich, B., Betancourt,
  M., Brubaker, M.~A., Guo, J., Li, P., and Riddell, A. (2017).
\newblock Stan: a probabilistic programming language.
\newblock {\em Grantee Submission}, 76(1):1--32.

\bibitem[Chen et~al., 2018]{Chen2018}
Chen, W., Chun, K.-J., and Barber, R.~F. (2018).
\newblock Discretized conformal prediction for efficient distribution-free
  inference.
\newblock {\em Stat}, 7(1):e173.

\bibitem[Chen et~al., 2016]{Chen2016}
Chen, W., Wang, Z., Ha, W., and Barber, R.~F. (2016).
\newblock Trimmed conformal prediction for high-dimensional models.
\newblock {\em arXiv preprint arXiv:1611.09933}.

\bibitem[Dawid, 1982]{Dawid1982}
Dawid, A.~P. (1982).
\newblock The well-calibrated {B}ayesian.
\newblock {\em Journal of the American Statistical Association},
  77(379):605--610.

\bibitem[Dua and Graff, 2017]{Dua2019}
Dua, D. and Graff, C. (2017).
\newblock {UCI} machine learning repository.

\bibitem[Dunn et~al., 2020]{Dunn2020}
Dunn, R., Wasserman, L., and Ramdas, A. (2020).
\newblock Distribution-free prediction sets with random effects.
\newblock {\em arXiv preprint arXiv:1809.07441v2}.

\bibitem[Efron et~al., 2004]{Efron2004}
Efron, B., Hastie, T., Johnstone, I., Tibshirani, R., et~al. (2004).
\newblock Least angle regression.
\newblock {\em Annals of statistics}, 32(2):407--499.

\bibitem[Fong and Holmes, 2020]{Fong2020}
Fong, E. and Holmes, C. (2020).
\newblock On the marginal likelihood and cross-validation.
\newblock {\em Biometrika}, 107(2):489--496.

\bibitem[Fraser et~al., 2011]{Fraser2011}
Fraser, D.~A. et~al. (2011).
\newblock Is {B}ayes posterior just quick and dirty confidence?
\newblock {\em Statistical Science}, 26(3):299--316.

\bibitem[Gammerman et~al., 1998]{Gammerman1998}
Gammerman, A., Vovk, V., and Vapnik, V. (1998).
\newblock Learning by transduction.
\newblock In {\em Proceedings of the Fourteenth conference on Uncertainty in
  artificial intelligence}, pages 148--155.

\bibitem[Gelman et~al., 2013]{Gelman2013}
Gelman, A., Carlin, J.~B., Stern, H.~S., Dunson, D.~B., Vehtari, A., and Rubin,
  D.~B. (2013).
\newblock {\em Bayesian data analysis}.
\newblock CRC press.

\bibitem[Gelman and Hill, 2006]{Gelman2006}
Gelman, A. and Hill, J. (2006).
\newblock {\em Data analysis using regression and multilevel/hierarchical
  models}.
\newblock Cambridge university press.

\bibitem[Gneiting and Raftery, 2007]{Gneiting2007}
Gneiting, T. and Raftery, A.~E. (2007).
\newblock Strictly proper scoring rules, prediction, and estimation.
\newblock {\em Journal of the American Statistical Association},
  102(477):359--378.

\bibitem[Harrison~Jr and Rubinfeld, 1978]{Harrison1978}
Harrison~Jr, D. and Rubinfeld, D.~L. (1978).
\newblock Hedonic housing prices and the demand for clean air.
\newblock {\em Journal of environmental economics and management},
  5(1):81--102.

\bibitem[Jansen, 2013]{Jansen2013}
Jansen, L. (2013).
\newblock {\em Robust {B}ayesian inference under model misspecification}.
\newblock PhD thesis, Master’s thesis, Leiden University.

\bibitem[Lei, 2019]{Lei2019}
Lei, J. (2019).
\newblock Fast exact conformalization of the lasso using piecewise linear
  homotopy.
\newblock {\em Biometrika}, 106(4):749--764.

\bibitem[Lei et~al., 2018]{Lei2018}
Lei, J., G’Sell, M., Rinaldo, A., Tibshirani, R.~J., and Wasserman, L.
  (2018).
\newblock Distribution-free predictive inference for regression.
\newblock {\em Journal of the American Statistical Association},
  113(523):1094--1111.

\bibitem[Little et~al., 2008]{Little2008}
Little, M., McSharry, P., Hunter, E., Spielman, J., and Ramig, L. (2008).
\newblock Suitability of dysphonia measurements for telemonitoring of
  parkinson’s disease.
\newblock {\em Nature Precedings}, pages 1--1.

\bibitem[Little, 2006]{Little2006}
Little, R.~J. (2006).
\newblock Calibrated {B}ayes: a {B}ayes/frequentist roadmap.
\newblock {\em The American Statistician}, 60(3):213--223.

\bibitem[Magnusson et~al., 2019]{Magnusson2019}
Magnusson, M., Andersen, M., Jonasson, J., and Vehtari, A. (2019).
\newblock Bayesian leave-one-out cross-validation for large data.
\newblock In {\em International Conference on Machine Learning}, pages
  4244--4253. PMLR.

\bibitem[Melluish et~al., 2001]{Melluish2001}
Melluish, T., Saunders, C., Nouretdinov, I., and Vovk, V. (2001).
\newblock Comparing the {B}ayes and typicalness frameworks.
\newblock In {\em European Conference on Machine Learning}, pages 360--371.
  Springer.

\bibitem[Pedregosa et~al., 2011]{Pedregosa2011}
Pedregosa, F., Varoquaux, G., Gramfort, A., Michel, V., Thirion, B., Grisel,
  O., Blondel, M., Prettenhofer, P., Weiss, R., Dubourg, V., Vanderplas, J.,
  Passos, A., Cournapeau, D., Brucher, M., Perrot, M., and Duchesnay, E.
  (2011).
\newblock Scikit-learn: Machine learning in {P}ython.
\newblock {\em Journal of Machine Learning Research}, 12:2825--2830.

\bibitem[Romano et~al., 2019]{Romano2019}
Romano, Y., Patterson, E., and Cand{\`e}s, E.~J. (2019).
\newblock Conformalized quantile regression.
\newblock {\em arXiv preprint arXiv:1905.03222}.

\bibitem[Salvatier et~al., 2016]{Salvatier2016}
Salvatier, J., Wiecki, T.~V., and Fonnesbeck, C. (2016).
\newblock Probabilistic programming in {P}ython using {PyMC}3.
\newblock {\em {PeerJ} Computer Science}, 2:e55.

\bibitem[Seabold and Perktold, 2010]{Seabold2010}
Seabold, S. and Perktold, J. (2010).
\newblock Statsmodels: Econometric and statistical modeling with python.
\newblock In {\em Proceedings of the 9th Python in Science Conference},
  volume~57, page~61. Austin, TX.

\bibitem[Shafer and Vovk, 2008]{Shafer2008}
Shafer, G. and Vovk, V. (2008).
\newblock A tutorial on conformal prediction.
\newblock {\em Journal of Machine Learning Research}, 9(Mar):371--421.

\bibitem[Tibshirani and Foygel, 2019]{Tibshirani2019}
Tibshirani, R. and Foygel, R. (2019).
\newblock Conformal prediction under covariate shift.
\newblock {\em Advances in neural information processing systems}.

\bibitem[Vehtari et~al., 2017]{Vehtari2017}
Vehtari, A., Gelman, A., and Gabry, J. (2017).
\newblock Practical {B}ayesian model evaluation using leave-one-out
  cross-validation and waic.
\newblock {\em Statistics and computing}, 27(5):1413--1432.

\bibitem[Vehtari et~al., 2015]{Vehtari2015}
Vehtari, A., Simpson, D., Gelman, A., Yao, Y., and Gabry, J. (2015).
\newblock Pareto smoothed importance sampling.
\newblock {\em arXiv preprint arXiv:1507.02646}.

\bibitem[Vovk, 2015]{Vovk2015}
Vovk, V. (2015).
\newblock Cross-conformal predictors.
\newblock {\em Annals of Mathematics and Artificial Intelligence}, 74(1):9--28.

\bibitem[Vovk et~al., 2005]{Vovk2005}
Vovk, V., Gammerman, A., and Shafer, G. (2005).
\newblock {\em Algorithmic learning in a random world}.
\newblock Springer Science \& Business Media.

\bibitem[Wasserman, 2011]{Wasserman2011}
Wasserman, L. (2011).
\newblock Frasian inference.
\newblock {\em Statistical Science}, pages 322--325.

\bibitem[Wolberg and Mangasarian, 1990]{Wolberg1990}
Wolberg, W.~H. and Mangasarian, O.~L. (1990).
\newblock Multisurface method of pattern separation for medical diagnosis
  applied to breast cytology.
\newblock {\em Proceedings of the national academy of sciences},
  87(23):9193--9196.

\bibitem[Zeni et~al., 2020]{Zeni2020}
Zeni, G., Fontana, M., and Vantini, S. (2020).
\newblock Conformal prediction: a unified review of theory and new challenges.
\newblock {\em arXiv preprint arXiv:2005.07972}.

\end{thebibliography}

\newpage
\appendix
\section{Bayesian and Frequentist Intervals}
Bayesian predictive intervals are conditioned on the specific observed sequence $Z_{1:n}$ and make statements on the next value  $[Y_{n+1} \mid X_{n+1}]$. In contrast, a conformal (frequentist) interval relates to the properties of intervals returned by the algorithm if run repeatedly across different data sets of size $n$. Subjective Bayesian statements on predictions are non-refutable, and are in this sense unscientific, but are optimal according to decision theoretic foundations. Meanwhile, frequentist statements are in principle verifiable and hence refutable.

To us, in the hands of an expert analyst with careful prior elicitation, the Bayesian conditional argument is the more persuasive for posterior and predictive uncertainty. The Bayesian predictive provides statements of uncertainty conditional on what has been observed, and so decisions pertain to each specific dataset. However, to make such strong statements, the Bayesian must usually make the strict assumption of the model being well-specified. If we wish to ensure that the predictive coverage of reported intervals is calibrated on average across repeats under weaker assumptions, then the conformal intervals are much more suitable. More details contrasting probability and confidence can be found in \citet[Section 2.2]{Shafer2008}.

At the end of the day, the Bayes and frequentist answer different questions, and the common confusion arises when treating them as answering the same. As long as we are aware they are addressing different needs, we believe both solutions are informative and useful, and indeed that is our recommendation in this paper. 

\section{Derivation of IS weights for Hierarchical Models}
We provide a quick derivation for the importance weights in Section \ref{sec:group_conf} to estimate
\begin{equation*}
p(Y_{i,j} \mid X_{i,j}, \bar{Z}_y) = \int f_{\theta_j,\tau}(Y_{i,j} \mid X_{i,j})\,  \pi(\theta_j,\tau \mid \bar{Z}_y) \, d\theta_j \, d\tau 
\end{equation*}
from posterior samples $[\theta_{1:J}^{(1:T)}, \tau^{(1:T)}, \phi^{(1:T)}] \sim \pi(\theta_{1:J},\tau,\phi \mid Z)$. We can write the above as
\begin{equation*}
\begin{aligned}
p(Y_{i,j} \mid X_{i,j}, \bar{Z}_y) =&\int f_{\theta_j,\tau}(Y_{i,j} \mid X_{i,j})\,  \pi(\theta_{1:J},\phi,\tau \mid \bar{Z}_y) \, d\theta_{1:J}\, d\phi \, d\tau  \\
&= \int f_{\theta_j,\tau}(Y_{i,j} \mid X_{i,j})\,  \frac{\pi(\theta_{1:J},\phi,\tau \mid \bar{Z}_y)}{\pi(\theta_{1:J},\phi,\tau \mid Z)}\pi(\theta_{1:J},\phi,\tau \mid Z) \, d\theta_{1:J} d\phi \, d\tau 
\end{aligned}
\end{equation*}
where
\begin{equation*}
    \frac{\pi(\theta_{1:J},\phi,\tau \mid \bar{Z}_y)}{\pi(\theta_{1:J},\phi,\tau \mid Z)} \propto f_{\theta_j,\tau}(y \mid X_{n_j+1,j}).
\end{equation*}
As the importance weight only depends on $\theta_j,\tau$, we only require the marginal posterior samples $[\theta_j^{(1:T)}, \tau^{(1:T)}]$ and we have that
\begin{equation*}
\begin{aligned}
    \hat{p}(Y_{i,j} \mid X_{i,j}, \bar{Z}_y) = \sum_{t = 1}^T w^{(t)}f_{\theta_j^{(t)}, \tau^{(t)}}(Y_{i,j}\mid X_{i,j})\\
w^{(t)} = f_{\theta_j^{(t)},\tau^{(t)}}(y \mid X_{n_j+1,j}), \quad
\widetilde{w}^{(t)} = \frac{w^{(t)}}{\sum_{t'=1}^T w^{(t')}} \cdot
\end{aligned}
\end{equation*}
\newpage
\section{Datasets, Licenses and Societal Impact}
We demonstrate our examples on 5 datasets, namely the the diabetes dataset  \citep{Efron2004}, the Boston housing dataset \citep{Harrison1978}, the Wisconsin Breast cancer dataset \citep{Wolberg1990}, the Parkinson's dataset \citep{Little2008} and the Radon dataset \citep{Gelman2006}. The first 3 datasets are available in \texttt{sklearn}, the Parkinson's dataset can be found on the UCI machine learning repository\footnote{\href{https://archive.ics.uci.edu/ml/datasets/Parkinson's}{\url{https://archive.ics.uci.edu/ml/datasets/Parkinson's}}} \citep{Dua2019}, and the Radon dataset is available on Andrew Gelman's website\footnote{\href{http://www.stat.columbia.edu/~gelman/arm/examples/radon/}{\url{http://www.stat.columbia.edu/~gelman/arm/examples/radon/}}}. Details on data acquisition is provided in the relevant references. We verified that the datasets do not contain personally identifiable information or offensive content by manual checking. The package \texttt{sklearn} is distributed under the 3-Clause BSD license. JAX and PyMC3 are both distributed under the Apache License, V2. 

The conformal method relies on the weak assumption of exchangeability. In terms of negative societal impacts, it may be tempting to apply the method blindly to real world problems without challenging this assumption as it seems quite weak. Applications where calibration is very important but data is not exchangeable would then be at risk.

\section{Additional Experiments}
\subsection{Experimental Details}
For all experiments, we repeat train-test splits or simulations 50 times, where the 70-30 train-test splits are random. For each repeat, we compute the average coverage and lengths for the test set. Means and standard errors are then computed from the 50 test set average coverages/lengths. For all MCMC examples, we generate $T = 8000$ samples, with $4000$ tune steps for sparse regression/classification and $8000$ for the hierarchical example.

\subsection{Sparse Regression}
\subsubsection{Diabetes}
We repeat analysis on the diabetes dataset, but this time with priors
\begin{equation}\label{eq:priors2}
\begin{aligned}
    f_\theta(y \mid x) &= \mathcal{N}(y \mid \theta^\T x + \theta_0, \tau^2)\\
    \pi(\theta_j) = \text{Normal}(0,d), \quad \pi(\theta_0) &= \text{Normal}(0,d),  \quad \pi(\tau)  = \mathcal{N}^+(0,1)
    \end{aligned}
\end{equation}
where we have different values $d = 5, 0.001$ which corresponds to weak and strong regularization towards 0.  For the baselines, we instead use ridge regression, with and without cross-validation for split/full as before. We emphasize that it is not exactly a fair comparison for the $d = 0.001$ case, as the baselines are tuning the parameter $\lambda$, whereas CB is subject to the misspecified prior. We still include them as baselines however, but highlight that they are not affected by the value of $d$.

MCMC required 30.8s and 13.2s for $d = 5$ and $d = 0.001$ respectively. The effect on coverage of setting $d = 0.001$ is not as detrimental as before as seen in Table \ref{tab:diab2}, as the posterior on $\tau$ compensates by increasing in value; the posterior mean is $\bar{\tau} = 1$ for $d = 0.001$ versus $\bar{\tau} = 0.71$ for $d =5$.

\begin{table}[!h]
\begin{center}
\caption{Diabetes; Coverage values \textit{not} within 3 standard errors (in brackets) of the target coverage $(1-\alpha) = 0.8$ are in {\color{red}\textbf{red}}.}
\label{tab:diab2}
\begin{tabular}{c c c c c c}
\hline
&& Bayes & CB & Split&Full ($\lambda =0.004$)\\
\hline 
 Coverage &$d = 5$&0.805 (0.005)&0.809 (0.005) &  0.816 (0.006)&0.809 (0.005)\\ 
 &$d = 0.001$&{\color{red}{\textbf{0.779}}} (0.006) & {0.809} (0.006) & /&/\\
    \hline
 Length &$d = 5$& 1.85 (0.01) &1.86 (0.01) & 1.94 (0.02)&1.86 (0.01)\\ 
 &$d = 0.001$& 2.56 (0.01)& 2.60 (0.01) &  /& / \vspace{0.5mm}\\
    \hline
 \rule{0pt}{0.8\normalbaselineskip} Run-time & $d = 1$& 0.417 (0.002)&0.677 (0.003)&    0.024 (0.000)&8.409 (0.007)\\
(secs)  &$d =0.02$ &0.540 (0.116)&0.692 (0.008)&  /&/\\
\end{tabular}
\end{center}
\end{table}

\newpage
\subsubsection{Boston Housing}
The Boston housing dataset \citep{Harrison1978} is of size $n = 506$, consisting of $d = 13$ predictors relating to housing such as demographic and air quality, with the response as the median value of  owner-occupied homes. 
We use the same Bayesian model as in (\ref{eq:priors}), again considering $c = 1,0.02$. For $c = 1$, the model is already misspecified for the Boston housing dataset as the errors are non-normal and have heavy tails \citep{Jansen2013}. All experimental settings are the same as in Section \ref{sec:sparse_reg}.\\

MCMC required an average of $22.8$s and $24.4$s for $c = 1,0.02$ to produce $T = 8000$ posterior samples. Again, in Table \ref{tab:boston} we see similar behaviour to the diabetes dataset case, but we note that even under $c = 1$, the Bayesian model over-covers. This is likely due to the presence of heavy tails in the residuals, leading to more conservative Bayesian predictive intervals. Here for $c = 1$, CB attains very close to nominal coverage and has a noticeably smaller average length. For $c=0.02$, CB is not affected much but the Bayes interval under-covers.

\begin{table}[!h]
\begin{center}
\caption{Boston; Coverage values \textit{not} within 3 standard errors (in brackets) of the target coverage $(1-\alpha) = 0.8$ are in {\color{red}\textbf{red}}.}
\label{tab:boston}
\begin{tabular}{c c c c c c}
\hline
&& Bayes & CB & Split & Full ($\lambda = 0.004$) \\
\hline
 Coverage &$c = 1$& {\color{red}\textbf{0.860}} (0.004) & 0.800 (0.005) & 0.805 (0.006) &0.799 (0.005)\\ 
 &$c = 0.02$&{\color{red}\textbf{0.728}} (0.005)& 0.799 (0.005)& /&/\vspace{0.5mm}\\
    \hline \vspace{0.5mm}
 Length &$c = 1$& 1.35 (0.01) & 1.12 (0.01) & 1.22 (0.02)& 1.12 (0.01)\\
 &$c = 0.02$ &0.96 (0.00) & 1.13 (0.01) & /& / \vspace{0.5mm}\\
    \hline \vspace{0.5mm}
  \rule{0pt}{0.8\normalbaselineskip} Run-time &$c = 1$ & 0.414 (0.003)&0.746 (0.011)&0.061 (0.000)&12.448 (0.042)\\
 (secs)  &$c =0.02$ & 0.406 (0.003)& 0.744 (0.003)& /&/
\end{tabular}
\end{center}
\end{table}
\subsection{Sparse Classification}

\subsubsection{Parkinson's Disease}
We provide an additional demonstration on the Parkinson's dataset \citep{Little2008}, which consists of $n =195$ voice recordings (after removing missing data) of patients with or without Parkinson's disease encoded in the binary response. The covariates consist of $d = 22$ different voice recording properties.\\

The experimental setup is identical to Section \ref{sec:sparse_class}, and MCMC required 29.2s to produce $T = 8000$ samples. Again, in Table \ref{tab:park} we see that Bayes over-covers even for reasonable priors, and CB produces tighter intervals that are closer to nominal coverage.

\begin{table}[!h]
\begin{center}
\caption{Parkinson's; Coverage values \textit{not} within 3 standard errors (in brackets) of the target coverage $(1-\alpha) = 0.8$ are in {\color{red}\textbf{red}}. ``Size'' denotes the average number of elements in the conformal prediction set, averaged over the test points and repetitions.}
\label{tab:park}
\begin{tabular}{c c c c c}
 \hline
& Bayes & CB &  Split&Full\\
\hline
 Coverage & {\color{red}\textbf{0.955}} (0.004) & 0.815 (0.008)&{\color{red}\textbf{0.842}} (0.010)&0.816 (0.008) \vspace{0.5mm}\\
    \hline \vspace{0.5mm}
 Size & 1.31 (0.01) & 0.93 (0.01) &1.05 (0.02)&0.95 (0.01)\vspace{0.5mm}\\
    \hline \vspace{0.5mm}
  \rule{0pt}{0.8\normalbaselineskip} Times &  0.203 (0.003)&  0.379 (0.008)& 0.478 (0.057)&0.168 (0.003)\\
\end{tabular}
\end{center}
\end{table}
\newpage
\subsubsection{Uninformative Predictions}

As the Bayesian model returns $p := p(y =1 \mid x, Z_{1:n})$, we compute $(1-\alpha)$ predictive sets by returning the smallest set of $\{0\},\{1\},\{0,1\}$ such that it contains at least $(1-\alpha)$ of the predictive probability mass. In other words, we return:
\begin{equation}
\begin{cases}
\{0\} \quad &\text{if }(1-p) \geq (1-\alpha)\\
\{1\}\quad &\text{if } p \geq (1-\alpha)\\
\{0,1\}\quad &\text{if } \max\{(1-p),p\} \leq (1-\alpha)\\
\end{cases}
\end{equation}
As this process is quite conservative, it is unsuprising that Bayes overcovers. The set $\{0,1\}$ is clearly uninformative, as it is always correct. On the other hand, the conformal sets can take on the empty set $\{\}$ as well, which we know to be incorrect. As discussed in \cite{Melluish2001}, empty set predictions correspond to being unable to make a prediction at the desired confidence level. \cite{Shafer2008} discusses the notion of \textit{confidence} and \textit{credibility}, which correspond to the greatest $(1-\alpha)$ such that the conformal set is of size 1 and the greatest $\alpha$ such that the conformal set is empty respectively.

We can decompose the informative and uninformative predictive sets and look at the misclassification rate, which is the error percentage within single element predictions. In comparison, we can look at the percentage of uninformative predictives (either both elements or empty). This is shown in Tables \ref{tab:misclass_rates}, \ref{tab:uninf_rates}, where the target coverage is $(1-\alpha) =0.8$ as before. For the breast cancer dataset, CB has a very small misclassification rate, but almost 19\% of all prediction sets are empty and 0\% are both, so the coverage is attained by making either single correct predictions or empty ones. Bayes on the other consists of more misclassifications but fewer uninformative predictions, but the attained coverage is a much higher value of 0.99. For the Parkinson's dataset, CB makes very few uninformative predictions, but has a relatively high misclassification rate.  Bayes on the other hand is very conservative, with 31\% uninformative predictions, hence the high average length and over-coverage. It is interesting to note the two sorts of behaviours attained by CB, which likely depends on the Bayesian model that was used to construct the CB intervals. 

In Figure \ref{fig:test_p}, we see the distributions of $p_i:= p(y_i \mid x_i,Z)$ of the Bayesian model with the corresponding CB interval length. We see that for CB intervals of length 1, the values of $p_i$ tend to be heavily skewed towards $0$ or $1$, which corresponds to the Bayesian model being strongly predictive. For empty CB intervals, in both cases the probability mass is distributed away from $0$ and $1$; for the breast cancer dataset it is evenly distributed on $(0,1)$ whereas for Parkinson's it is concentrated around $0.5$. When given a CB interval of length 0, it may be more informative to actually return the value $p_i$, which is the corresponding Bayesian prediction. 

\begin{table}[!h]
\begin{center}
\caption{Misclassification rates}\label{tab:misclass_rates}
\begin{tabular}{c   c  c }
\hline
Dataset &Bayes & CB \\
\hline 
 Breast Cancer  & 0.011 (0.001)& 0.002 (0.000) \\
 Parkinson's & 0.064 (0.006)&0.124 (0.004) \\
\end{tabular}
\end{center}
\end{table}

\begin{table}[!h]
\begin{center}
\caption{Uninformative Rates}\label{tab:uninf_rates}
\begin{tabular}{c   c c cc}
&\multicolumn{2}{c}{Both}&\multicolumn{2}{c}{Empty}\\
\hline
Dataset  &Bayes & CB&Bayes & CB \\
\hline 
 Breast Cancer  &0.059 (0.002) & 0.000 (0.000)& 0 &0.186 (0.005) \\
 Parkinson's  & 0.312 (0.009)&0.003 (0.001)&0&0.070 (0.006)\\
\end{tabular}
\end{center}
\end{table}

 \begin{figure}[!h]
\centering
 \includegraphics[width=0.95\textwidth]{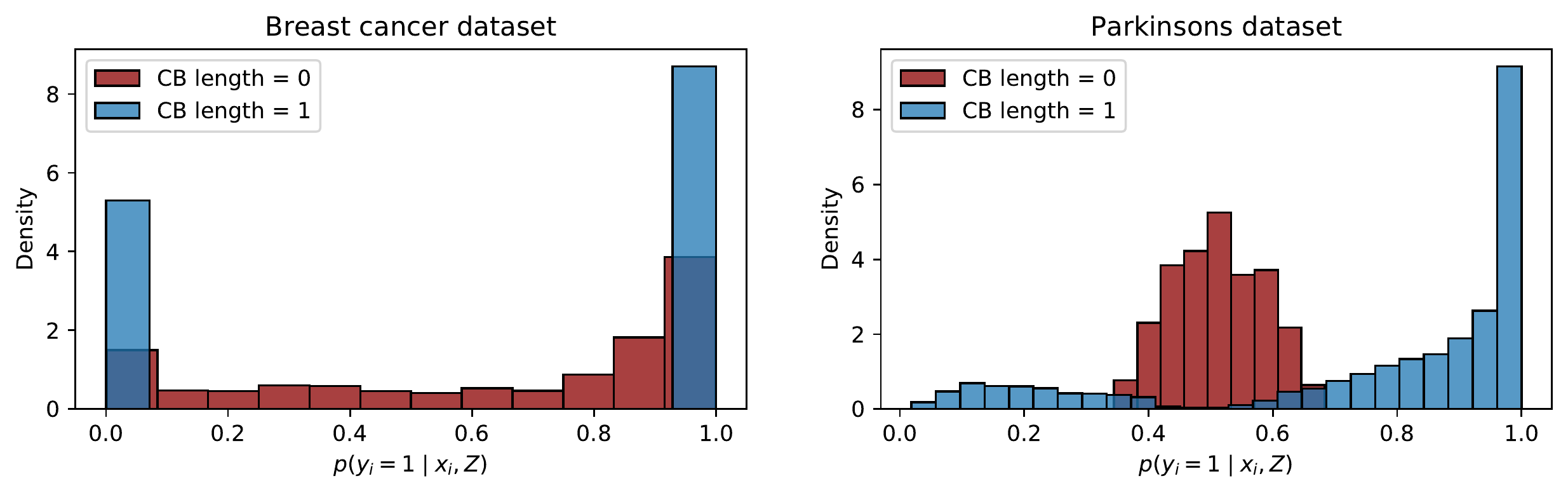}
\caption{Distribution of $p(y_i \mid x_i,Z)$ on test data for CB intervals of length $1$ or $0$ for breast cancer (left) and Parkinson's (right).}\label{fig:test_p}
\end{figure}

\newpage
\subsection{Hierarchical}
\subsubsection{Radon} 
Using the same model as in (\ref{eq:hier_model}), we analyze the radon dataset\footnote{We base this example on the PyMC3 notebook here: \href{https://docs.pymc.io/notebooks/multilevel_modeling.html}{\url{https://docs.pymc.io/notebooks/multilevel_modeling.html/}}}, introduced in  \citet[Chapter~12]{Gelman2006}. The dataset consists of 919 home radon levels in Minnesota, where the covariate is the location of measurement, with $x = 0$ corresponding to basement and $x = 1$ to the first floor. The groups are the 85 counties in which the homes are located, and vary significantly in group size. Around half of the counties contain $n_j \leq 4$ measurements, with the smallest county containing one value and the largest containing 119.

As many of the group sizes are quite small, we do not repeat train-test splits and evaluate coverage. Instead, we compare the CB and Bayes intervals on the entire dataset for different floor values $x$ and counties, and discuss the effects of $n_j$ on the choice of $\alpha_j$. 
As each $x \in \{0,1\}$, we specify $x_{\text{test}}$ as all possible group indicators and predictors, resulting in $85 \times 2 = 170$ test values. For the predictive intervals, we use a grid of size 100 between [-6,6]. MCMC for the radon example required around 156s, and computing the 170 predictive intervals took 0.65s and 2.69s for Bayes and CB respectively, where we have excluded the first run compilation time for JAX.

As we need $\alpha_j \geq {1}/({n_j+1})$ to get intervals that are not the entire real line, we set $\alpha_j = {1.1}/({n_j+1})$ (for numerical reasons) and compare the Bayes and CB intervals. The average CB length is 2.66 compared to 2.17 for the Bayes intervals, noting that we are averaging over all possibilities instead of the distribution of $x_{\text{test}}$. In Figure \ref{fig:radon_big}, we plot $\pi_j(y)$ for the two value of $x \in \{0,1\}$ for two groups. For the group size $n_j = 4$, we see that $\pi_j(y) \geq 0.2$, so any $\alpha < 0.2$ would return us the real line as the confidence set. For $n_j = 52$, the ranks are much smoother, giving us more resolution in the confidence sets with respect to $\alpha$. In Figure \ref{fig:radon_small}, we show the rank plots for $n_j = 1$, which only contain the ranks $\{0.5,1\}$. Interestingly, for county 41, $x = 1$ returns the empty set for $\alpha \geq 0.5$ and the real line for $\alpha < 0.5$, which is a consequence of the small group size. CB is able to return non empty sets for county 49 with $\alpha \geq 0.5$. All CB sets appear to be connected.

 \begin{figure}[!h]
\centering
 \includegraphics[width=0.95\textwidth]{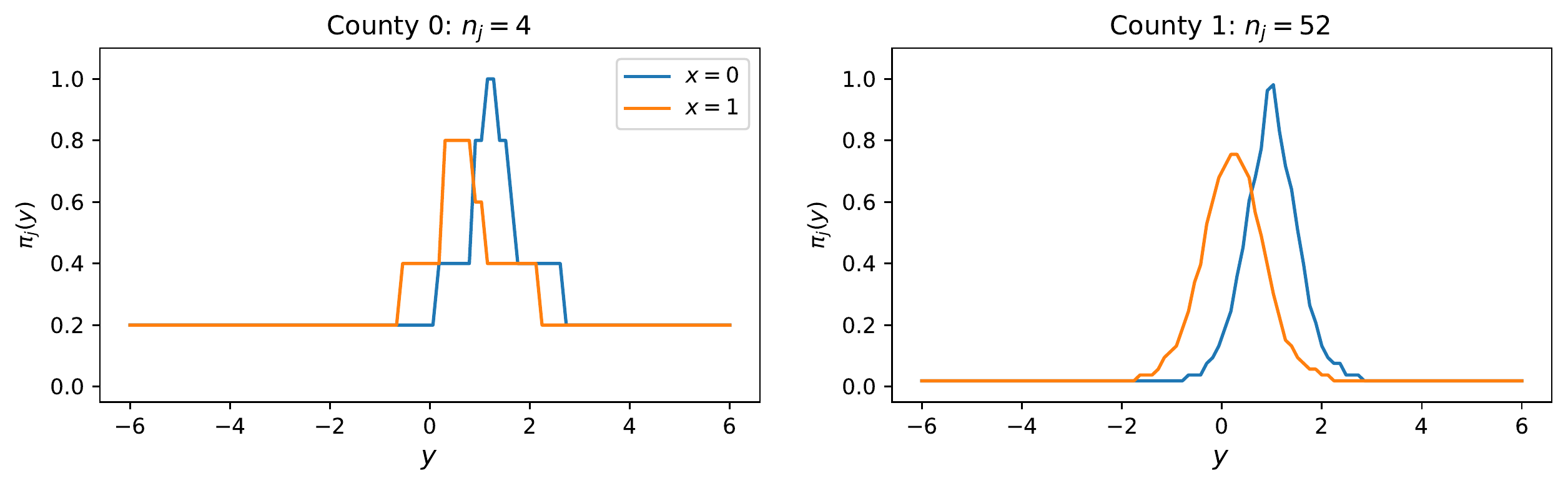}
\caption{Plot of rank $\pi_j(y)$ for $x \in \{0,1\}$ with $n_j = 4$ (left) and $n_j = 52$ (right).}\label{fig:radon_big}
\end{figure}

 \begin{figure}[!h]
\centering
 \includegraphics[width=0.95\textwidth]{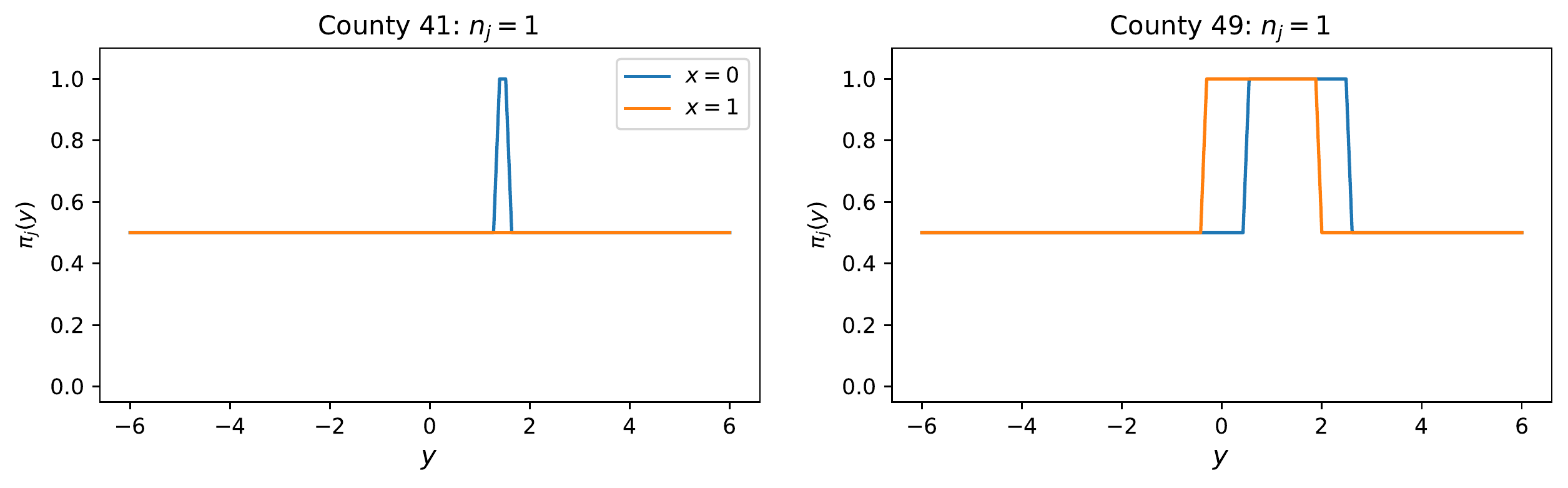}
\caption{Plot of rank $\pi_j(y)$ for $x \in \{0,1\}$ with $n_j = 1$ for two groups.}\label{fig:radon_small}
\end{figure}

As a reference, fitting a linear mixed-effects model in \texttt{statsmodels} \citep{Seabold2010} to the whole dataset takes around 600ms, so the full conformal method, which would require refitting for each of the 100 grid value and 170 test values, would require 170 minutes in total. On the other hand, CB requires much less time and has similar group structure.

\newpage
\subsection{MCMC Times}
In Table \ref{tab:mcmc}, we report the average times (and standard errors) for running MCMC on the NC6 virtual machine. We point out that the times for the hierarchical methods are longer as we needed to increase the tuning steps and acceptance probability to prevent divergences in the chains. 

\begin{table}[!h]
\begin{center}
\caption{Run-time in seconds for MCMC}\label{tab:mcmc}
\begin{tabular}{c   c }
\hline
Dataset & MCMC \\
\hline 
 Diabetes $(c = 1)$ \hspace{3.5mm} & 21.868 (0.135)\\ 
 Diabetes $(c =0.02)$& 26.790 (0.365) \\
 Diabetes $(d = 5)$&30.825 (0.214)\\
 Diabetes $(d = 0.001)$&13.166 (0.072)\\
 Boston  $(c = 1)$ \hspace{3.5mm} & 22.827 (0.036)  \\  
 Boston $(c =0.02)$ &24.362 (0.429)\\
 \\
 Breast Cancer &  45.418 (0.804)\\
 Parkinson's &29.239 (0.302) \\
 \\
  Scenario 1 & 90.109 (1.605)\\
 Scenario 2 &78.403 (1.544) \\
 Radon &156.087 (0.000) 
\end{tabular}
\end{center}
\end{table}
\newpage
\subsection{Grid effects}
To quantify the grid effects, we also compute the coverage by directly evaluating $\pi(Y_{n+1})$ for each test point and checking if it satisfies condition (\ref{eq:conformal_set}). Of course in practice this is not possible as we do not observe $Y_{n+1}$.

 For the grid conformal method, we compute the $y \in \mathcal{Y}_{\text{grid}}$ that is nearest to $Y_{n+1}$, and report $0$ or $1$ if this grid value is in the conformal prediction set. Note that this implementation of the grid method can both under and over cover. Denote $\delta$ as the resolution of the grid, and the smallest grid value in the conformal prediction set as $a$. If $a - \delta <Y_{n+1}< a - \delta/2$, we may incorrectly reject $Y_{n+1}$ if it is truly in the set and $a-\delta$ is not. Similarly, if $a - \delta/2 <Y_{n+1}< a$ we can incorrectly accept if $Y_{n+1}$ is not actually in the set but $a$ is. Note that the estimated average length is also affected by this.

We compare the grid and exact method in Tables \ref{tab:diab_grid}, \ref{tab:boston_grid}, \ref{tab:hier_grid}. The largest discrepancy in average coverage is only $0.008$, which is quite negligible. However, we expect this discrepancy to increase as $|\mathcal{Y}_{\text{grid}}|$ decreases.

\begin{table}[!h]
\begin{center}
\caption{Diabetes; Grid versus exact coverage, with target $(1-\alpha)= 0.8$}\label{tab:diab_grid}
\begin{tabular}{c c c c }
\hline
&& CB Grid & CB  Exact \\
\hline
 Coverage & $c = 1$ &0.808 (0.006)&0.810 (0.005)\\ 
 & $c= 0.02$ &0.809 (0.006)&0.810 (0.006)\\ 
\end{tabular}
\end{center}
\end{table}

\begin{table}[!h]
\begin{center}
\caption{Boston;  Grid versus exact coverage, with target $(1-\alpha)= 0.8$}\label{tab:boston_grid}
\begin{tabular}{c c c c }
\hline
&& CB Grid & CB  Exact \\
\hline
 Coverage & $c = 1$ &0.800 (0.005) & 0.800 (0.005)\\ 
 & $c = 0.02$ &0.799 (0.005)& 0.799 (0.005)
\\ 
\end{tabular}
\end{center}
\end{table}

\begin{table}[!h]\label{tab:hier_sim_grid}
\begin{center}
\caption{Simulated grouped dataset;  Grid versus exact coverage, with target $(1-\alpha)= 0.8$}\label{tab:hier_grid}
\begin{tabular}{c c c c c  c c}
&&\multicolumn{2}{c}{Scenario 1}&\multicolumn{2}{c}{Scenario 2}\\
\hline
&Group& CB Grid & CB Exact &CB Grid & CB Exact\\
\hline \vspace{0.5mm}
 Coverage&1 & 0.794 (0.022)&0.786 (0.023)  & 0.786 (0.025)&0.790 (0.024) \\ 
    &2 &  0.812 (0.024)&0.816 (0.023)&  0.812 (0.024)& 0.818 (0.023) \\ 
    &3 & 0.824 (0.022)&0.820 (0.022)& 0.824 (0.020) &0.824 (0.020)\\ 
    &4 & 0.798 (0.022)&0.796 (0.021)& 0.836 (0.021)&0.838 (0.022) \\ 
    &5 & 0.810 (0.020)&0.812 (0.019)& 0.796 (0.022)& 0.792 (0.022)  \vspace{0.5mm}\\ 
\end{tabular}
\end{center}
\end{table}

\end{document}